\documentclass[aip,jmp,preprint,amssymb,a4paper]{revtex4-1}
\usepackage{bbm}
\usepackage{enumerate}
\usepackage{color}
\usepackage{amsmath}
\usepackage{amssymb}
\usepackage[amsmath, thmmarks]{ntheorem}
\usepackage{graphicx}

\theorembodyfont{\normalfont}
\newtheorem{dfn}{Definition}

\theoremsymbol{\ensuremath{\scriptstyle\blacksquare}}

\theoremstyle{nonumberplain}
\theoremheaderfont{\itshape}
\newtheorem{proof}{Proof}

\theoremclass{LaTeX}
\newtheorem{thm}[dfn]{Theorem}
\newtheorem{theorem}[dfn]{Theorem}

\newtheorem{corollary}[dfn]{Corollary}

\def\scaled{\let\onleft=\left\let\onright=\right}
\def\unscale{\let\onleft=\relax\let\onright=\relax}
\unscale
\newcommand{\hilbert}[1]{\ensuremath{\mathfrak{#1}}}

\newcommand{\bbone}{\ensuremath{\mathbbm 1}}

\newcommand{\conj}{\ensuremath{^*}}

\newcommand{\defeq}{\ensuremath{:=}}
\newcommand{\setdef}{\ensuremath{\;\vert\;}}

\newcommand{\ocmpl}{\ensuremath{^{\mathrm{c}}}}

\def\scaled{\let\onleft=\left\let\onright=\right}
\def\unscale{\let\onleft=\relax\let\onright=\relax}

\newcommand{\set}[1]{\ensuremath{\onleft\{ #1\onright\}}\unscale}

\newcommand{\notperp}{\mathbin{\perp\kern-.9em/}}
\newcommand{\compat}{\mathbin{\leftrightarrow}}
\newcommand{\notsim}{\mathbin{\sim\kern-.9em/}}

\newcommand{\powerset}[1]{2^{#1}}

\def\quotient#1#2{%
    \raise1ex\hbox{$#1$}\Big/\lower1ex\hbox{$#2$}%
}

\begin{document}

\title{Remarks on the tensor product structure of no-signaling theories}
\author{Tomasz I. Tylec}
\email{tomasz.tylec@ug.edu.pl}
\affiliation{Institute of Theoretical Physics and Astrophysics,
  University of Gdańsk,
  ul. Wita Stwosza 57
  80-308 Gdańsk}
\affiliation{Center for Theoretical Physics,
  Polish Academy of Sciences,
  Aleja Lotnik\'ow 32/46,
  02-668 Warsaw, Poland}
\author{Marek Ku\'s}
\affiliation{Center for Theoretical Physics,
  Polish Academy of Sciences,
  Aleja Lotnik\'ow 32/46,
  02-668 Warsaw, Poland}

\begin{abstract}
  In the quantum logic framework we show that the
  no-signaling box model is a particular type of
  tensor product of the logics of single boxes.
  Such notion of tensor product is too strong
  to apply in the category of logics of
  quantum mechanical systems.
  Consequently, we show that the no-signaling box models
  cannot be considered as
  generalizations of quantum mechanical models.
\end{abstract}

\maketitle

\section{Introduction}

Let us consider the following very simple model: a composite system
consisting of two parties, each of them is a ``black box'', a device
that produce an output value $\alpha$ when it is provided with an input
value $a$. An input can be physically identified with an observable on
the box, while the output is the outcome of that measurement. Any such
device is characterized by the input set and a family of sets, indexed
by input values, of allowed outputs. It is assumed that all these sets
are finite \cite{barrett2005nonlocal,barrett2007information}. A state
of such device defines a probability $P(\alpha|a)$ of getting $\alpha$
given $a$. The system of two such devices produces a pair of output
values $(\alpha,\beta)$ upon pair of input values $(a, b)$. Again, the
state of composite system defines the probability $P(\alpha\beta|ab)$
of getting $(\alpha, \beta)$ given $(a, b)$. Additionally, we assume
certain independence conditions, called a \emph{no-signaling} hold:
\begin{equation*}
  \sum_{\alpha} P(\alpha\beta|ab) = \sum_\alpha P(\alpha\beta|cb),\qquad
  \sum_{\beta} P(\alpha\beta|ab) = \sum_\beta P(\alpha\beta|ac).
\end{equation*}
This is an expression of Einstein's causality principle that forbids
instantaneous interactions (thus also information transfer); in plain
words: ``what happens in one box does not influence the other.''
Intuitively, a no-signaling model is a composite system in which all
components are coupled in minimal physically sensible way.

The above described models were introduced by Popescu and
Rohrlich\cite{popescu1994quantum} (thus their alternative name is
Popescu-Rohrlich or PR-boxes) as an example of a system violating
quantum bound of Bell-type inequalities, while still being compatible
with Einstein's causality principle. Since then, such models found many
applications in quantum information theory, ranging from security of
communication and distributed computing~\cite{
  PhysRevLett.97.120405,Buhrman2006,Short2006},
communication complexity~\cite{PhysRevLett.96.250401,
  Buhrman2010,Linden2007,Brunner2009},
to quantifying randomness~\cite{Pironio2010,Gallego2013}.
It is also a widely used tool in discussions related to
foundations of physics~\cite{Pawlowski:2009aa,
  Masanes2011,Oppenheim2010,Fritz2013,Navascues2009,Barnum2010a,Allcock2009}.

Despite numerous applications, a rigorous mathematical treatment of
no-signaling models is rather scarce. The need for such is two-fold.
Firstly, since there are no physical realizations of no-signaling
models, properties of such systems cannot be verified experimentally
but only on the basis of rigorous mathematical framework. Secondly,
since no-signaling models actually generalize probability beyond
quantum probability, we should expect new qualitative changes, in
analogy to passing from classical to quantum probability. Only a
systematic study of mathematical structure can reveal such new
non-intuitive properties.

As far as the authors are concerned, up to date there are only two
rigorous mathematical treatments of no-signaling theories. One is based
on the convex set approach and is usually called a \emph{Generalized}
or \emph{Generic Probability Theory (GPT)}
(see~Ref.~\onlinecite{Barnum2006} and references therein). The second
approach is due to us~\cite{tylec2015,tylec2015-2} and describes
no-signaling models in the framework of quantum logics (in the sense
of~Ref.~\onlinecite{ptak1991orthomodular}). Ours approach is less
general than GPT, but covers standard definition of no-signaling models
and allows to study more fine-grained structures of the theory.

One of the greatest arguments against quantum logic framework is the
problem of defining suitable tensor product. From that point of view
the convex set approach seems to be better, as the notion of tensor
products (since they are not unique; see
e.g.~Ref.~\onlinecite{namioka1969tensor} or
Ref.~\onlinecite{Jarchow1981}) is well studied for them. It can be
shown that the state space of no-signaling model is a maximal tensor
product (as a tensor product of ordered linear spaces) of state spaces
of its components, while the subset of classically correlated states is
a minimal tensor product. In this paper we show how the quantum logics
of components of no-signaling model combine together to form a quantum
logic of the whole system. Due to the simple structure of the
components, there exists a suitable notion of tensor product of quantum
logics \cite{pulmannova1983} that fits this scheme.

\section{Quantum structures}
\label{sec:quantum-structures}

Let us recall the most important definitions and facts
that will be used in the sequel.

\begin{dfn}
  \label{thm:def-qlogic}
  A \emph{quantum logic} is a partialy ordered set $L$
  with a map $\ocmpl\colon L\to L$ such that
  \begin{enumerate}
  \item[L1] there exists the greatest (denoted by $\bbone$)
    and the least (denoted by $0$) element in $L$,
  \item[L2] map $p \mapsto p\ocmpl$ is order reversing, i.e.\
    $p \le q$ implies that $q\ocmpl \le p\ocmpl$,
  \item[L3] map $p \mapsto p\ocmpl$ is idempotent, i.e.\
    $(p\ocmpl)\ocmpl = p$,
  \item[L4] for a countable family $\set{p_i}$, s.t.\ $p_i \le p_j\ocmpl$
    for $i\neq j$, the supremum $\bigvee \set{p_i}$ exists,
  \item[L5] if $p\le q$ then $q = p \vee(q\wedge p\ocmpl)$
    (orthomodular law),
  \end{enumerate}
  where
  $p\vee q$ is the least upper bound
  and $p \wedge q$ the greatest lower bound of $p$ and $q$.
\end{dfn}

Two elements $p, q$ of quantum logic $L$ are called \emph{disjoint}
whenever $p \le q\ocmpl$. An element $p$ is said to \emph{cover} $q$
whenever $q \le r \le p$ implies $r=q$ or $r=p$. Elements covering $0$
are called \emph{atoms} and $L$ is called \emph{atomistic} whenever any
element $q\in L$ is a supremum of all atoms less than $q$. In a typical
way we define a \emph{sublogic} $K$ of a quantum logic $L$ as a subset
$K\subset L$ closed under orthocompletion and countable sums of
disjoint elements.

\begin{dfn}
  A \emph{state} $\rho$ on a quantum logic $L$ is a map $\rho\colon L
  \to [0, 1]$, s.t.\
  \begin{enumerate}
  \item[S1] $\rho(\bbone) = 1$,
  \item[S2] for a countable family $\set{p_i}$, s.t.\ $p_i \le p_j\ocmpl$
    $\rho(\bigvee\set{p_i}) = \sum_i p_i$.
  \end{enumerate}
  We will be denoted by $\mathcal S(\mathcal L)$ the set of all states
  on a quantum logic $L$
\end{dfn}

\begin{dfn}
    Elements $p, q\in L$ of a quantum logic $L$ are \emph{compatible},
    what we denote by $p\compat q$, whenever there exist pairwise
    disjoint elements $p_1, q_1, r$ such that $p = p_1 \vee r, q = q_1
    \vee r$.

    More generally, a subset $A\subset L$ is said to be
    \emph{compatible} whenever for any finite subset $\set{p_1, \dots,
    p_n}\subset A$ there exist finite subset $G\subset L$, such that
    (i) elements of $G$ are mutually disjoint, (ii) any $p_i$ is
    supremum of some subset of $G$.
\end{dfn}

It is easier to think about compatibility in terms of the following
property:

\begin{theorem}[Ref.~\onlinecite{ptak1991orthomodular}, Thm. 1.3.23]
    Let $A\subset\mathcal L$ be a compatible subset of quantum logic.
    Then there exists a Boolean sublogic $\mathcal K\subset \mathcal L$,
    s.t. $A\subset \mathcal K$.\label{thm:compat}
\end{theorem}

For orthomodular lattices, compatibility of a set $A$ is equivalent to
pairwise compatibility of elements. For general quantum logic it is no
longer true. However, if quantum logic is \emph{regular}, i.e.\ for any
triple of mutually compatible elements $\set{a, b, c}$, $a \compat
b\vee c$, a set $A$ is compatible if and only if all elements are
pairwise compatible (see~Ref.~\onlinecite{ptak1991orthomodular}, Def.
1.3.26 and Prop. 1.3.27.)

\begin{dfn}[see~Ref.~\onlinecite{ptak1991orthomodular}, Sec. 1.1]
  Let $\Delta$ be a family of subsets of some set $\Omega$
  with partial order relation given by set inclusion
  and $A\ocmpl = \Omega\setminus A$ satisfying:
  \begin{enumerate}
  \item[C1] $\emptyset\in\Delta$,
  \item[C2] $A\in \Delta$ implies $\Omega\setminus A \in \Delta$,
  \item[C3] for any countable family $\set{A_i}\subset \Delta$
    of mutually disjoint sets
    $\bigcup \set{A_i} \in \Delta$.
  \end{enumerate}
  Then $(\Omega, \Delta)$ is called a concrete (quantum) logic.
\end{dfn}

\begin{dfn}[see Def.\ 44 in Ref.~\onlinecite{pulmanova2007}]
    Let $\mathcal L$ be a quantum logic.
    The set of states $\mathcal S$ is said to
    be \emph{rich} whenever:
    \begin{equation*}
        \set{\mu \in \mathcal S\setdef \mu(a) = 1} \subset
        \set{\mu \in \mathcal S\setdef \mu(b) = 1} \qquad\implies\qquad
        a \le b.
    \end{equation*}
    We say that $\mathcal L$ is \emph{rich} whenever it has a rich
    subset of states. \label{thm:rich}
\end{dfn}

\begin{thm}[see Thm.\ 48 in Ref.~\onlinecite{pulmanova2007}]
    \label{thm:set-repr} A quantum logic $L$ is set-representable,
    i.e.\ there exists order preserving isomorphism between $L$ and
    some concrete logic $(\Omega, \Delta)$, if and only if $L$ has a
    rich set of two-valued states (i.e.\ states with a property that
    $\forall q\in L, \sigma(q) = 1$ or $\sigma(q)=0$).
\end{thm}

\begin{dfn}[see Ref.~\onlinecite{varadarajan2007geometry} p. 53]
  Let $S\subset\mathcal S$ be subset of the set of states $\mathcal S$ of a
  quantum logic $L$. We say that $\mu\in\mathcal S$ is \emph{a superposition} of
  states in $S$ whenever
  \begin{equation*}
    \nu(a) = 0,\forall \nu\in S
    \qquad \implies \qquad
    \mu(a) = 0.
  \end{equation*}
  Let us denote by $\overline{S} = \set{\mu \in \mathcal S\setdef \mu \text{ is
      superposition of states in } S}$.
\end{dfn}

From the physical point of view, it is essential to have a tool to
describe composite systems. There are some arguments in the direction
that it should always be a kind of (categorical) tensor
product\cite{aerts1978physical}. Unfortunately, this notion presents
some difficulties in the theory of quantum logics when one tries to
define it universally (cf. Ref.~\onlinecite{dvurecenskij2000new},
chapter 4, and references therein). Let us recall only these notions
that we will discuss in the sequel. The following definition is a
direct generalization of Def.~3 of Ref.~\onlinecite{matolcsi1975tensor}
to a category of regular quantum logics.

\begin{dfn}
    Let $L_1, L_2$ be regular quantum logics.
    The \emph{free orthodistributive product} of $L_1, L_2$
    is a triple $(L, u_1, u_2)$,
    where $L$ is a quantum logic and
    \begin{enumerate}[(i)]
        \item $u_i\colon L_i\to L$ are monomorphisms,
        \item $u_1(L_1) \cup u_2(L_2)$ generates $L$,
        \item $u_1(a) \wedge u_2(b) = 0$ iff $a = 0$ or $b = 0$,
        \item $u_1(a) \compat u_2(b)$ for any $a\in L_1, b \in L_2$.
    \end{enumerate}
\end{dfn}

Pulmannov\'a\cite{pulmannova1985tensor} has shown that the existence of
free orthodistributive product in the category of atomistic
$\sigma$-lattices is quite an exceptional case:
\begin{theorem}[Ref.~\onlinecite{pulmannova1985tensor}, Thm.~2]
    Let $L, L_1, L_2$ be complete atomistic $\sigma$-orthomodular lattices,
    and let $(L, u_1, u_2)$ be free orthodistributive product of $L_1, L_2$.
    Let $A, A_1, A_2$ be atoms of $L, L_1, L_2$ respectively.
    If
    \begin{equation}
        A = \set{u_1(a)\wedge u_2(b)\setdef a \in A_1, b \in A_2}
        \label{eq:atoms-gen}
    \end{equation}
    then at least one of $L_1, L_2$ is a Boolean algebra.
    \label{thm:free-ortho-trivial}
\end{theorem}
Let us remark that the property~\eqref{eq:atoms-gen} is not satisfied
for a pair of lattices of projections on a Hilbert space, thus this
theorem is not applicable to the tensor product of Hilbert spaces
(which obviously is a free orthodistributive product of lattices of
projections).

In order to define \emph{a tensor product} of quantum logics (which is
required to be a free orthodistributive product from the
category-theoretical standpoint) we want to specify not only how the
set of propositions behave but also how sets of states combine
together. \hbox{Pulmannov\'a} proposed the following two definitions
\cite{pulmannova1983,pulmannova1985tensor}:

\begin{dfn}
    Let $L, K$ be quantum logics with $\mathcal S, \mathcal R$ being a
    state spaces of $L$ and $K$ respectively. A quantum logic $T$ with
    a state space $\mathcal U$ will be called a \emph{strong tensor
    product} of $L$ and $K$ whenever there are mappings $\alpha\colon
    L\times K \to T$, $\beta\colon \mathcal S, \mathcal R\to U$ such
    that:
    \begin{enumerate}[(i)]
        \item $\beta(\mu, \nu)(\alpha(a, b)) = \mu(a)\nu(b)$,
        \item set $\beta(\mathcal S, \mathcal R)$ is rich for $T$,
        \item $T$ is generated by $\alpha(L, K)$.
    \end{enumerate}
    if instead of (ii) following two weaker conditions are satisfied
    \begin{enumerate}
        \item[(ii')] $\set{\chi\in\mathcal U\setdef \chi(c)=1} =
            \overline{\set{\beta(\mu, \nu)\setdef \beta(\mu, \nu)(c) = 1}}$,
            for all $c$ of the form:
            \begin{enumerate}[(a)]
                \item $c = \wedge_i \alpha(a_i, b_i)$, for all $a_i\in L, b_i\in K$ or
                \item $c = \alpha(a, \bbone)$ for all $a \in L$ or
                \item $c = \alpha(\bbone, b)$ for all $b \in L$,
            \end{enumerate}
        \item[(ii'')] $\overline{\beta(\mathcal S, \mathcal R)} = \mathcal U$
    \end{enumerate}
    then we say that $T$ is a \emph{weak tensor product} of $L$ and $K$.
\end{dfn}

From the previously quoted Thm.~\ref{thm:free-ortho-trivial} it
follows\cite{pulmannova1985tensor} that in the category of atomistic
$\sigma$-orthomodular lattices the strong tensor product exists only if
at least one of the components is a Boolean algebra. In particular, it
does not exist for lattices of projections on separable Hilbert spaces.
On the other hand, if $L_1, L_2$ are lattices of projections on Hilbert
spaces $\hilbert H_1, \hilbert H_2$, then the weak tensor product
exists and coincides with the lattice of projections on $\hilbert
H_1\otimes \hilbert H_2$ (see Ref.~\onlinecite{pulmannova1985tensor},
Sec. IV; actually, there is second, inequivalent possibility - the
lattice of projections on $\hilbert H_1\conj\otimes \hilbert H_2$, but
this is not relevant for our considerations here).

% add definition of tensor product in the category of D-posets/effect algebras.
% add definition of 0-1 pasting

\begin{dfn}[cf. Ref. \onlinecite{ptak1991orthomodular}]
    Let $\set{L_i}_i$ be a countable family of quantum logics.
    A \emph{0-1-pasting} of $\set{L_i}$ is quantum logics $L$
    defined as quotient of disjoint union of $L_i$'s:
    \begin{equation*}
        L = \coprod_i L_i / \sim,
    \end{equation*}
    where $a \sim b$ iff $a, b$ are both the least elements
    or the greatest elements in any of $L_i$'s.
    In other words, $L$ is the disjoint union of logics $L_i$
    glued togheter at $0$ and $\bbone$.
\end{dfn}

\section{Logic of non-signalling boxes as a tensor product}
\label{sec:tensor-product}

Let us briefly recall the structure of the logic of arbitrary two box
system; cf. Ref.~\onlinecite{tylec2015-2,tylec2015} for detailed
discussion and proof that the construction below indeed yields the
logic of arbitrary no-signaling box model. To fix the notation, let the
first box accept $N$ distinct inputs, labelled by $1, \dots, N$. For
each input $a=1,\dots, N$ let $\mathcal U_a$ denote the set of possible
outcomes (also of finite cardinality). Denote by $\mathcal U =
(\mathcal U_1, \dots, \mathcal U_N)$. Similarly, let the second box
accept $M$ distinct inputs, again labelled by $1,\dots, M$ and set of
outcomes for intput $b$ will be denoted by $\mathcal V_b$. Let
$\mathcal V = (\mathcal V_1, \dots, \mathcal V_M)$. We call such two
box system a $(\mathcal U, \mathcal V)$-box world.

Denote by $\Gamma_1 = \set{(x_1, \dots, x_N)\setdef x_a\in\mathcal
U_a}$ and $\Gamma_2 = \set{(y_1, \dots, y_M)\setdef y_b\in\mathcal
\mathcal V_b}$ the classical phase space that can be associated with
the first and the second box, respectively. The composite (classical)
system will be described a classical phase space $\Gamma =
\Gamma_1\times \Gamma_2$. An experimental question ``does pair of
inputs $(a, b)$ results in a pair of outputs $(\alpha, \beta)$'' can be
associated with the following subset of the phase space $\Gamma$:
\begin{equation*}
    [a \alpha, b \beta] \defeq \set{(x, y)\in \Gamma\setdef
        x_a = \alpha, y_b = \beta}
\end{equation*}
The logic $L$ of $(\mathcal U, \mathcal V)$-box world is the quantum
logic generated in the Boolean algebra $\powerset\Gamma$ by all
questions of the above form (see~Ref.~\onlinecite{tylec2015-2,
tylec2015} for details):
\begin{equation*}
     A = \set{ [a \alpha, b \beta] \setdef
        (x, y)\in \Gamma\setdef x_a = \alpha, y_b = \beta}
\end{equation*}
Of course, $L$ does not have to be (and is not) a Boolean algebra
anymore.

In a similar way we can assign a logic $L_{\mathcal U}$ to the first
box. It is the quantum logic $L_{\mathcal U}$ generated in
$\powerset{\Gamma_1}$ by elements $[a \alpha] \defeq \set{x \in
\Gamma_1\setdef x_a = \alpha}$. Observe that $L_{\mathcal U}$ is a
0-1-pasting of the family of Boolean logics $\set{\powerset{\mathcal
U_a}}_{a=1,\dots, N}$. Intuitievly speaking, this means that we do not
impose any relations between outputs for different inputs. In the same
manner we define logic $L_{\mathcal V}$ of the second box.

\begin{theorem}
  The quantum logic $L$ of $(\mathcal U, \mathcal V)$-box system is an
  orthodistributive product of logics $L_{\mathcal U}$ and $L_{\mathcal
  V}$.
\end{theorem}
\begin{proof}
  Since $L_{\mathcal U}$ is a 0-1 pasting of Boolean algebras
  $\powerset{\mathcal U_i}$, any nonzero element of $L_{\mathcal U}$ is
  of the form:
  \begin{equation*}
    [a \in A] = \bigoplus_{\alpha\in A} [a \alpha]
  \end{equation*}
  where $a=1,\dots,N$ and $A\subset \mathcal U_a$.
  The same is true for $L_{\mathcal V}$.
  Maps
  \begin{align*}
    u\colon L_{\mathcal U} \to L,
    &\qquad u([a \in A]) =
      \set{(x, y)\subset \Gamma \
      \setdef x\in A} \equiv [a \in A, \bbone],\\
    v\colon L_{\mathcal V} \to L,
    &\qquad v([b \in B]) =
      \set{(x, y)\subset \Gamma
      \setdef y\in B} \equiv [\bbone, b \in B]
  \end{align*}
  are clearly injective mappings from $L_{\mathcal U}$ and $L_{\mathcal
  V}$ to $L$. Since any atom $[a \alpha, b \beta]$ of $L$ equals to $[a
  \alpha, \bbone] \wedge [\bbone, b \beta]$, union of images
  $u(L_{\mathcal U})\cup v(L_{\mathcal V})$ generates $L$. Clearly
  $u([a\in A])\wedge v([b \in B])$ exists and in
  Ref.~\onlinecite{tylec2015-2} we have shown that $u([a \in A])\compat
  v([b \in B])$.
\end{proof}

\begin{corollary}
  Thm.~\ref{thm:free-ortho-trivial} is no longer valid in the category
  of regular quantum logics.
\end{corollary}

Actually, we can show even more:

\begin{theorem}
    The logic $L$ of $(\mathcal U, \mathcal V)$-box world is a strong
    tensor product of single box logics $L_1, L_2$.
    \label{thm:tensor-prod}
\end{theorem}
\begin{proof}
    Let
    \begin{equation*}
        \Phi([a\in A], [b\in B]) = [a\in A, b\in B] =
        \bigoplus_{\alpha\in A, \beta\in B} [a \alpha, b \beta]
    \end{equation*}
    and $\Psi(\mu, \nu)([a \alpha, b \beta]) = \mu(a \alpha) \nu(b
    \beta)$ (and extend by orthogonal sums to all $L$). It is clear
    that
    \begin{equation*}
        \Psi(\mu, \nu)(\alpha([a\in A, b\in B])) = \mu([a\in A])\nu([b\in B]).
    \end{equation*}
    and $L$ is generated (by definition) by elements $\phi([a\in A],
    [b\in B])$. It remains to show that the set $\Psi(\mathcal
    S_{\mathcal U}, \mathcal S_{\mathcal V})$, where $\mathcal S_i$ is
    set of states on $L_i$, is rich for $L$.

    Firstly, observe that since $L$ is a concrete logic, it has a rich
    set of two-valued states. Any such state $\chi$ on $L$ can be
    characterized by the set of atoms $\mathcal O_\chi$ on which it
    obtains value $1$ (logic is atomistic, so the value of a state on
    atoms fully describe the state). Moreover, if $[p, q], [r, s] \in
    \mathcal O_\chi$, then also $[p, s], [r, q] \in \mathcal O_\chi$.
    Indeed, it follows from
    \begin{equation*}
        1 = \mu([p, q]) = \mu([p, q]) + \mu([p\ocmpl, q]) =
        \mu([\bbone, q]) = \mu([r, q]) + \mu([r\ocmpl, q]) = 1
    \end{equation*}
    and
    \begin{equation*}
        1 = \mu([r, s]) = \mu([r, s]) + \mu([r, s\ocmpl]) =
        \mu([r, \bbone]) = \mu([r, q]) + \mu([r, q\ocmpl]) = 1
    \end{equation*}
    that either $\mu([r, q)]) = 1$ or $\mu([r\ocmpl, s]) = 1 = \mu([r,
    s\ocmpl])$. But the latter cannot be true since $[r\ocmpl, s]\perp
    [r, s\ocmpl]$. The proof that $[p, s]\in \mathcal O_\chi$ is completely analogous.
    Consequently, we can select subsets $\mathcal O_\chi^1, \mathcal
    O_\chi^2$ of atoms in $L_1$ and $L_2$ respectively, such that:
    \begin{equation*}
        \mathcal O_\chi = \set{[p, q] \setdef
            p\in \mathcal O_\chi^1,
            q\in \mathcal O_\chi^2}
    \end{equation*}

    In the next step we define two functions $\mu, \nu$ on $L_1$ and
    $L_2$ respectively, by:
    \begin{align*}
        \mu(p) &= 1 \text{ if $p\in\mathcal O_\chi^1$, otherwise } 0,\\
        \nu(q) &= 1 \text{ if $q\in\mathcal O_\chi^2$, otherwise } 0,
    \end{align*}
    on atoms $p, q$ in $L_1, L_2$
    and extended by orthogonal sums to the whole $L_1, L_2$.
    Clearly $\chi([p, q]) = \mu(p)\nu(q)$
    for any $p\in \mathcal L_1$ and $q\in\mathcal L_2$.
    Moreover, since there can be no disjoint pair
    in $\mathcal O_\chi^1$ (otherwise, $\chi$ would not be a state),
    we infer that $\mu$ is a state on $L_1$.
    The same is true for $\nu$.
    Consequently, we have shown that any two-valued state on $L$
    is a product state of two two-valued states on single box logics.
    Thus the set of product states is rich in $L$.
\end{proof}

This fact has remarkable consequences. As it was mentioned in the
introduction, for lattices of projections on Hilbert spaces, strong
tensor product does not exist. It is a consequence of the fact that
pure entangled states on the composite quantum system are not mixtures
but superpositions of pure product states~\cite{pulmannova1983}. Thus
to describe properly a composed quantum mechanical systems we have to
use weaker notion of the weak tensor product.

One the other hand, the logic of two non-signalling boxes is a strong
tensor product of logics of single boxes, so, contrary to the quantum
mechanical states, the set of product states fully describes the
physical structure of a box system. It might suggest that what is
called ``entanglement'' or ``non-local'' property of certain states on
box-world system is in fact a weaker notion than the quantum mechanical
entanglement (despite the fact that it allows stronger violation of
Bell-type inequalities). Giving precise meaning to this statement is an
interesting topic of a further research.

Our example suggests also that no-signaling box models
are not the best tools to investigate the question
of what distinguishes quantum mechanics from other no-signaling theories%
\cite{Pawlowski:2009aa}.
Their super-quantum properties are the result of a rather trivial structure,
allowing for strong tensor product to exists.
By no means one can state that the no-signaling boxes are more general than
the quantum mechanics, even restricted to the finite dimensional Hilbert spaces.

\begin{acknowledgements}
  This work was done with the support of John Templeton Foundation grant.
\end{acknowledgements}

\bibliography{library}

\begin{thebibliography}{32}%
\makeatletter
\providecommand \@ifxundefined [1]{%
 \@ifx{#1\undefined}
}%
\providecommand \@ifnum [1]{%
 \ifnum #1\expandafter \@firstoftwo
 \else \expandafter \@secondoftwo
 \fi
}%
\providecommand \@ifx [1]{%
 \ifx #1\expandafter \@firstoftwo
 \else \expandafter \@secondoftwo
 \fi
}%
\providecommand \natexlab [1]{#1}%
\providecommand \enquote  [1]{``#1''}%
\providecommand \bibnamefont  [1]{#1}%
\providecommand \bibfnamefont [1]{#1}%
\providecommand \citenamefont [1]{#1}%
\providecommand \href@noop [0]{\@secondoftwo}%
\providecommand \href [0]{\begingroup \@sanitize@url \@href}%
\providecommand \@href[1]{\@@startlink{#1}\@@href}%
\providecommand \@@href[1]{\endgroup#1\@@endlink}%
\providecommand \@sanitize@url [0]{\catcode `\\12\catcode `\$12\catcode
  `\&12\catcode `\#12\catcode `\^12\catcode `\_12\catcode `\%12\relax}%
\providecommand \@@startlink[1]{}%
\providecommand \@@endlink[0]{}%
\providecommand \url  [0]{\begingroup\@sanitize@url \@url }%
\providecommand \@url [1]{\endgroup\@href {#1}{\urlprefix }}%
\providecommand \urlprefix  [0]{URL }%
\providecommand \Eprint [0]{\href }%
\providecommand \doibase [0]{http://dx.doi.org/}%
\providecommand \selectlanguage [0]{\@gobble}%
\providecommand \bibinfo  [0]{\@secondoftwo}%
\providecommand \bibfield  [0]{\@secondoftwo}%
\providecommand \translation [1]{[#1]}%
\providecommand \BibitemOpen [0]{}%
\providecommand \bibitemStop [0]{}%
\providecommand \bibitemNoStop [0]{.\EOS\space}%
\providecommand \EOS [0]{\spacefactor3000\relax}%
\providecommand \BibitemShut  [1]{\csname bibitem#1\endcsname}%
\let\auto@bib@innerbib\@empty
%</preamble>
\bibitem [{\citenamefont {Barrett}\ \emph {et~al.}(2005)\citenamefont
  {Barrett}, \citenamefont {Linden}, \citenamefont {Massar}, \citenamefont
  {Pironio}, \citenamefont {Popescu},\ and\ \citenamefont
  {Roberts}}]{barrett2005nonlocal}%
  \BibitemOpen
  \bibfield  {author} {\bibinfo {author} {\bibfnamefont {J.}~\bibnamefont
  {Barrett}}, \bibinfo {author} {\bibfnamefont {N.}~\bibnamefont {Linden}},
  \bibinfo {author} {\bibfnamefont {S.}~\bibnamefont {Massar}}, \bibinfo
  {author} {\bibfnamefont {S.}~\bibnamefont {Pironio}}, \bibinfo {author}
  {\bibfnamefont {S.}~\bibnamefont {Popescu}}, \ and\ \bibinfo {author}
  {\bibfnamefont {D.}~\bibnamefont {Roberts}},\ }\bibfield  {title} {\enquote
  {\bibinfo {title} {{Nonlocal correlations as an information-theoretic
  resource}},}\ }\href@noop {} {\bibfield  {journal} {\bibinfo  {journal}
  {Physical Review A}\ }\textbf {\bibinfo {volume} {71}},\ \bibinfo {pages}
  {22101} (\bibinfo {year} {2005})}\BibitemShut {NoStop}%
\bibitem [{\citenamefont {Barrett}(2007)}]{barrett2007information}%
  \BibitemOpen
  \bibfield  {author} {\bibinfo {author} {\bibfnamefont {J.}~\bibnamefont
  {Barrett}},\ }\bibfield  {title} {\enquote {\bibinfo {title} {{Information
  processing in generalized probabilistic theories}},}\ }\href {\doibase
  10.1103/PhysRevA.75.032304} {\bibfield  {journal} {\bibinfo  {journal}
  {Physical Review A}\ }\textbf {\bibinfo {volume} {75}},\ \bibinfo {pages}
  {32304} (\bibinfo {year} {2007})},\ \Eprint {http://arxiv.org/abs/0508211}
  {arXiv:0508211 [quant-ph]} \BibitemShut {NoStop}%
\bibitem [{\citenamefont {Popescu}\ and\ \citenamefont
  {Rohrlich}(1994)}]{popescu1994quantum}%
  \BibitemOpen
  \bibfield  {author} {\bibinfo {author} {\bibfnamefont {S.}~\bibnamefont
  {Popescu}}\ and\ \bibinfo {author} {\bibfnamefont {D.}~\bibnamefont
  {Rohrlich}},\ }\bibfield  {title} {\enquote {\bibinfo {title} {{Quantum
  nonlocality as an axiom}},}\ }\href@noop {} {\bibfield  {journal} {\bibinfo
  {journal} {Foundations of Physics}\ }\textbf {\bibinfo {volume} {24}},\
  \bibinfo {pages} {379--385} (\bibinfo {year} {1994})}\BibitemShut {NoStop}%
\bibitem [{\citenamefont {Ac{\'{\i}}n}, \citenamefont {Gisin},\ and\
  \citenamefont {Masanes}(2006)}]{PhysRevLett.97.120405}%
  \BibitemOpen
  \bibfield  {author} {\bibinfo {author} {\bibfnamefont {A.}~\bibnamefont
  {Ac{\'{\i}}n}}, \bibinfo {author} {\bibfnamefont {N.}~\bibnamefont {Gisin}},
  \ and\ \bibinfo {author} {\bibfnamefont {L.}~\bibnamefont {Masanes}},\
  }\bibfield  {title} {\enquote {\bibinfo {title} {{From Bell's Theorem to
  Secure Quantum Key Distribution}},}\ }\href {\doibase
  10.1103/PhysRevLett.97.120405} {\bibfield  {journal} {\bibinfo  {journal}
  {Physical Review Letters}\ }\textbf {\bibinfo {volume} {97}},\ \bibinfo
  {pages} {120405} (\bibinfo {year} {2006})},\ \Eprint
  {http://arxiv.org/abs/0510094} {arXiv:0510094 [quant-ph]} \BibitemShut
  {NoStop}%
\bibitem [{\citenamefont {Buhrman}\ \emph {et~al.}(2006)\citenamefont
  {Buhrman}, \citenamefont {Christandl}, \citenamefont {Unger}, \citenamefont
  {Wehner},\ and\ \citenamefont {Winter}}]{Buhrman2006}%
  \BibitemOpen
  \bibfield  {author} {\bibinfo {author} {\bibfnamefont {H.}~\bibnamefont
  {Buhrman}}, \bibinfo {author} {\bibfnamefont {M.}~\bibnamefont {Christandl}},
  \bibinfo {author} {\bibfnamefont {F.}~\bibnamefont {Unger}}, \bibinfo
  {author} {\bibfnamefont {S.}~\bibnamefont {Wehner}}, \ and\ \bibinfo {author}
  {\bibfnamefont {A.}~\bibnamefont {Winter}},\ }\bibfield  {title} {\enquote
  {\bibinfo {title} {{Implications of superstrong non-locality for
  cryptography}},}\ }\href {\doibase 10.1098/rspa.2006.1663} {\bibfield
  {journal} {\bibinfo  {journal} {Proceedings of the Royal Society A:
  Mathematical, Physical and Engineering Sciences}\ }\textbf {\bibinfo {volume}
  {462}},\ \bibinfo {pages} {1919--1932} (\bibinfo {year} {2006})}\BibitemShut
  {NoStop}%
\bibitem [{\citenamefont {Short}, \citenamefont {Gisin},\ and\ \citenamefont
  {Popescu}(2006)}]{Short2006}%
  \BibitemOpen
  \bibfield  {author} {\bibinfo {author} {\bibfnamefont {A.~J.}\ \bibnamefont
  {Short}}, \bibinfo {author} {\bibfnamefont {N.}~\bibnamefont {Gisin}}, \ and\
  \bibinfo {author} {\bibfnamefont {S.}~\bibnamefont {Popescu}},\ }\bibfield
  {title} {\enquote {\bibinfo {title} {{The Physics of No-Bit-Commitment:
  Generalized Quantum Non-Locality Versus Oblivious Transfer}},}\ }\href
  {\doibase 10.1007/s11128-006-0015-4} {\bibfield  {journal} {\bibinfo
  {journal} {Quantum Information Processing}\ }\textbf {\bibinfo {volume}
  {5}},\ \bibinfo {pages} {131--138} (\bibinfo {year} {2006})}\BibitemShut
  {NoStop}%
\bibitem [{\citenamefont {Brassard}\ \emph {et~al.}(2006)\citenamefont
  {Brassard}, \citenamefont {Buhrman}, \citenamefont {Linden}, \citenamefont
  {M{\'{e}}thot}, \citenamefont {Tapp}, \citenamefont {Falk},\ and\
  \citenamefont {Unger}}]{PhysRevLett.96.250401}%
  \BibitemOpen
  \bibfield  {author} {\bibinfo {author} {\bibfnamefont {G.}~\bibnamefont
  {Brassard}}, \bibinfo {author} {\bibfnamefont {H.}~\bibnamefont {Buhrman}},
  \bibinfo {author} {\bibfnamefont {N.}~\bibnamefont {Linden}}, \bibinfo
  {author} {\bibfnamefont {A.~A.}\ \bibnamefont {M{\'{e}}thot}}, \bibinfo
  {author} {\bibfnamefont {A.}~\bibnamefont {Tapp}}, \bibinfo {author}
  {\bibfnamefont {U.}~\bibnamefont {Falk}}, \ and\ \bibinfo {author}
  {\bibfnamefont {F.}~\bibnamefont {Unger}},\ }\bibfield  {title} {\enquote
  {\bibinfo {title} {{Limit on Nonlocality in Any World in Which Communication
  Complexity Is Not Trivial}},}\ }\href {\doibase
  10.1103/PhysRevLett.96.250401} {\bibfield  {journal} {\bibinfo  {journal}
  {Physical Review Letters}\ }\textbf {\bibinfo {volume} {96}},\ \bibinfo
  {pages} {250401} (\bibinfo {year} {2006})},\ \Eprint
  {http://arxiv.org/abs/0508042} {arXiv:0508042 [quant-ph]} \BibitemShut
  {NoStop}%
\bibitem [{\citenamefont {Buhrman}\ \emph {et~al.}(2010)\citenamefont
  {Buhrman}, \citenamefont {Cleve}, \citenamefont {Massar},\ and\ \citenamefont
  {{De Wolf}}}]{Buhrman2010}%
  \BibitemOpen
  \bibfield  {author} {\bibinfo {author} {\bibfnamefont {H.}~\bibnamefont
  {Buhrman}}, \bibinfo {author} {\bibfnamefont {R.}~\bibnamefont {Cleve}},
  \bibinfo {author} {\bibfnamefont {S.}~\bibnamefont {Massar}}, \ and\ \bibinfo
  {author} {\bibfnamefont {R.}~\bibnamefont {{De Wolf}}},\ }\bibfield  {title}
  {\enquote {\bibinfo {title} {{Nonlocality and communication complexity}},}\
  }\href {\doibase 10.1103/RevModPhys.82.665} {\bibfield  {journal} {\bibinfo
  {journal} {Reviews of Modern Physics}\ }\textbf {\bibinfo {volume} {82}},\
  \bibinfo {pages} {665--698} (\bibinfo {year} {2010})},\ \Eprint
  {http://arxiv.org/abs/0907.3584} {arXiv:0907.3584} \BibitemShut {NoStop}%
\bibitem [{\citenamefont {Linden}\ \emph {et~al.}(2007)\citenamefont {Linden},
  \citenamefont {Popescu}, \citenamefont {Short},\ and\ \citenamefont
  {Winter}}]{Linden2007}%
  \BibitemOpen
  \bibfield  {author} {\bibinfo {author} {\bibfnamefont {N.}~\bibnamefont
  {Linden}}, \bibinfo {author} {\bibfnamefont {S.}~\bibnamefont {Popescu}},
  \bibinfo {author} {\bibfnamefont {A.~J.}\ \bibnamefont {Short}}, \ and\
  \bibinfo {author} {\bibfnamefont {A.}~\bibnamefont {Winter}},\ }\bibfield
  {title} {\enquote {\bibinfo {title} {{Quantum nonlocality and beyond: limits
  from nonlocal computation.}}}\ }\href {\doibase
  10.1103/PhysRevLett.99.180502} {\bibfield  {journal} {\bibinfo  {journal}
  {Physical Review Letters}\ }\textbf {\bibinfo {volume} {99}},\ \bibinfo
  {pages} {180502} (\bibinfo {year} {2007})}\BibitemShut {NoStop}%
\bibitem [{\citenamefont {Brunner}\ and\ \citenamefont
  {Skrzypczyk}(2009)}]{Brunner2009}%
  \BibitemOpen
  \bibfield  {author} {\bibinfo {author} {\bibfnamefont {N.}~\bibnamefont
  {Brunner}}\ and\ \bibinfo {author} {\bibfnamefont {P.}~\bibnamefont
  {Skrzypczyk}},\ }\bibfield  {title} {\enquote {\bibinfo {title} {{Nonlocality
  distillation and postquantum theories with trivial communication
  complexity.}}}\ }\href {\doibase 10.1103/PhysRevLett.102.160403} {\bibfield
  {journal} {\bibinfo  {journal} {Physical Review Letters}\ }\textbf {\bibinfo
  {volume} {102}},\ \bibinfo {pages} {160403} (\bibinfo {year}
  {2009})}\BibitemShut {NoStop}%
\bibitem [{\citenamefont {Pironio}\ \emph {et~al.}(2010)\citenamefont
  {Pironio}, \citenamefont {Ac{\'{\i}}n}, \citenamefont {Massar}, \citenamefont
  {de~la Giroday}, \citenamefont {Matsukevich}, \citenamefont {Maunz},
  \citenamefont {Olmschenk}, \citenamefont {Hayes}, \citenamefont {Luo},
  \citenamefont {Manning},\ and\ \citenamefont {Monroe}}]{Pironio2010}%
  \BibitemOpen
  \bibfield  {author} {\bibinfo {author} {\bibfnamefont {S.}~\bibnamefont
  {Pironio}}, \bibinfo {author} {\bibfnamefont {A.}~\bibnamefont
  {Ac{\'{\i}}n}}, \bibinfo {author} {\bibfnamefont {S.}~\bibnamefont {Massar}},
  \bibinfo {author} {\bibfnamefont {A.~B.}\ \bibnamefont {de~la Giroday}},
  \bibinfo {author} {\bibfnamefont {D.~N.}\ \bibnamefont {Matsukevich}},
  \bibinfo {author} {\bibfnamefont {P.}~\bibnamefont {Maunz}}, \bibinfo
  {author} {\bibfnamefont {S.}~\bibnamefont {Olmschenk}}, \bibinfo {author}
  {\bibfnamefont {D.}~\bibnamefont {Hayes}}, \bibinfo {author} {\bibfnamefont
  {L.}~\bibnamefont {Luo}}, \bibinfo {author} {\bibfnamefont {T.~A.}\
  \bibnamefont {Manning}}, \ and\ \bibinfo {author} {\bibfnamefont
  {C.}~\bibnamefont {Monroe}},\ }\bibfield  {title} {\enquote {\bibinfo {title}
  {{Random numbers certified by Bell's theorem.}}}\ }\href {\doibase
  10.1038/nature09008} {\bibfield  {journal} {\bibinfo  {journal} {Nature}\
  }\textbf {\bibinfo {volume} {464}},\ \bibinfo {pages} {1021--4} (\bibinfo
  {year} {2010})}\BibitemShut {NoStop}%
\bibitem [{\citenamefont {Gallego}\ \emph {et~al.}(2013)\citenamefont
  {Gallego}, \citenamefont {Masanes}, \citenamefont {{De La Torre}},
  \citenamefont {Dhara}, \citenamefont {Aolita},\ and\ \citenamefont
  {Ac{\'{\i}}n}}]{Gallego2013}%
  \BibitemOpen
  \bibfield  {author} {\bibinfo {author} {\bibfnamefont {R.}~\bibnamefont
  {Gallego}}, \bibinfo {author} {\bibfnamefont {L.}~\bibnamefont {Masanes}},
  \bibinfo {author} {\bibfnamefont {G.}~\bibnamefont {{De La Torre}}}, \bibinfo
  {author} {\bibfnamefont {C.}~\bibnamefont {Dhara}}, \bibinfo {author}
  {\bibfnamefont {L.}~\bibnamefont {Aolita}}, \ and\ \bibinfo {author}
  {\bibfnamefont {A.}~\bibnamefont {Ac{\'{\i}}n}},\ }\bibfield  {title}
  {\enquote {\bibinfo {title} {{Full randomness from arbitrarily deterministic
  events.}}}\ }\href {\doibase 10.1038/ncomms3654} {\bibfield  {journal}
  {\bibinfo  {journal} {Nature Communications}\ }\textbf {\bibinfo {volume}
  {4}},\ \bibinfo {pages} {2654} (\bibinfo {year} {2013})}\BibitemShut
  {NoStop}%
\bibitem [{\citenamefont {Pawlowski}\ \emph {et~al.}(2009)\citenamefont
  {Pawlowski}, \citenamefont {Paterek}, \citenamefont {Kaszlikowski},
  \citenamefont {Scarani}, \citenamefont {Winter},\ and\ \citenamefont
  {Zukowski}}]{Pawlowski:2009aa}%
  \BibitemOpen
  \bibfield  {author} {\bibinfo {author} {\bibfnamefont {M.}~\bibnamefont
  {Pawlowski}}, \bibinfo {author} {\bibfnamefont {T.}~\bibnamefont {Paterek}},
  \bibinfo {author} {\bibfnamefont {D.}~\bibnamefont {Kaszlikowski}}, \bibinfo
  {author} {\bibfnamefont {V.}~\bibnamefont {Scarani}}, \bibinfo {author}
  {\bibfnamefont {A.}~\bibnamefont {Winter}}, \ and\ \bibinfo {author}
  {\bibfnamefont {M.}~\bibnamefont {Zukowski}},\ }\bibfield  {title} {\enquote
  {\bibinfo {title} {{Information causality as a physical principle}},}\ }\href
  {http://dx.doi.org/10.1038/nature08400} {\bibfield  {journal} {\bibinfo
  {journal} {Nature}\ }\textbf {\bibinfo {volume} {461}},\ \bibinfo {pages}
  {1101--1104} (\bibinfo {year} {2009})}\BibitemShut {NoStop}%
\bibitem [{\citenamefont {Masanes}\ and\ \citenamefont
  {M{\"{u}}ller}(2011)}]{Masanes2011}%
  \BibitemOpen
  \bibfield  {author} {\bibinfo {author} {\bibfnamefont {L.}~\bibnamefont
  {Masanes}}\ and\ \bibinfo {author} {\bibfnamefont {M.~P.}\ \bibnamefont
  {M{\"{u}}ller}},\ }\bibfield  {title} {\enquote {\bibinfo {title} {{A
  derivation of quantum theory from physical requirements}},}\ }\href {\doibase
  10.1088/1367-2630/13/6/063001} {\bibfield  {journal} {\bibinfo  {journal}
  {New Journal of Physics}\ }\textbf {\bibinfo {volume} {13}},\ \bibinfo
  {pages} {063001} (\bibinfo {year} {2011})}\BibitemShut {NoStop}%
\bibitem [{\citenamefont {Oppenheim}\ and\ \citenamefont
  {Wehner}(2010)}]{Oppenheim2010}%
  \BibitemOpen
  \bibfield  {author} {\bibinfo {author} {\bibfnamefont {J.}~\bibnamefont
  {Oppenheim}}\ and\ \bibinfo {author} {\bibfnamefont {S.}~\bibnamefont
  {Wehner}},\ }\bibfield  {title} {\enquote {\bibinfo {title} {{The uncertainty
  principle determines the nonlocality of quantum mechanics.}}}\ }\href
  {\doibase 10.1126/science.1192065} {\bibfield  {journal} {\bibinfo  {journal}
  {Science (New York, N.Y.)}\ }\textbf {\bibinfo {volume} {330}},\ \bibinfo
  {pages} {1072--4} (\bibinfo {year} {2010})}\BibitemShut {NoStop}%
\bibitem [{\citenamefont {Fritz}\ \emph {et~al.}(2013)\citenamefont {Fritz},
  \citenamefont {Sainz}, \citenamefont {Augusiak}, \citenamefont {Brask},
  \citenamefont {Chaves}, \citenamefont {Leverrier},\ and\ \citenamefont
  {Ac{\'{\i}}n}}]{Fritz2013}%
  \BibitemOpen
  \bibfield  {author} {\bibinfo {author} {\bibfnamefont {T.}~\bibnamefont
  {Fritz}}, \bibinfo {author} {\bibfnamefont {A.~B.}\ \bibnamefont {Sainz}},
  \bibinfo {author} {\bibfnamefont {R.}~\bibnamefont {Augusiak}}, \bibinfo
  {author} {\bibfnamefont {J.~B.}\ \bibnamefont {Brask}}, \bibinfo {author}
  {\bibfnamefont {R.}~\bibnamefont {Chaves}}, \bibinfo {author} {\bibfnamefont
  {A.}~\bibnamefont {Leverrier}}, \ and\ \bibinfo {author} {\bibfnamefont
  {A.}~\bibnamefont {Ac{\'{\i}}n}},\ }\bibfield  {title} {\enquote {\bibinfo
  {title} {{Local orthogonality as a multipartite principle for quantum
  correlations.}}}\ }\href {\doibase 10.1038/ncomms3263} {\bibfield  {journal}
  {\bibinfo  {journal} {Nature Communications}\ }\textbf {\bibinfo {volume}
  {4}},\ \bibinfo {pages} {2263} (\bibinfo {year} {2013})}\BibitemShut
  {NoStop}%
\bibitem [{\citenamefont {Navascues}\ and\ \citenamefont
  {Wunderlich}(2009)}]{Navascues2009}%
  \BibitemOpen
  \bibfield  {author} {\bibinfo {author} {\bibfnamefont {M.}~\bibnamefont
  {Navascues}}\ and\ \bibinfo {author} {\bibfnamefont {H.}~\bibnamefont
  {Wunderlich}},\ }\bibfield  {title} {\enquote {\bibinfo {title} {{A glance
  beyond the quantum model}},}\ }\href {\doibase 10.1098/rspa.2009.0453}
  {\bibfield  {journal} {\bibinfo  {journal} {Proceedings of the Royal Society
  A: Mathematical, Physical and Engineering Sciences}\ }\textbf {\bibinfo
  {volume} {466}},\ \bibinfo {pages} {881--890} (\bibinfo {year}
  {2009})}\BibitemShut {NoStop}%
\bibitem [{\citenamefont {Barnum}\ \emph {et~al.}(2010)\citenamefont {Barnum},
  \citenamefont {Beigi}, \citenamefont {Boixo}, \citenamefont {Elliott},\ and\
  \citenamefont {Wehner}}]{Barnum2010a}%
  \BibitemOpen
  \bibfield  {author} {\bibinfo {author} {\bibfnamefont {H.}~\bibnamefont
  {Barnum}}, \bibinfo {author} {\bibfnamefont {S.}~\bibnamefont {Beigi}},
  \bibinfo {author} {\bibfnamefont {S.}~\bibnamefont {Boixo}}, \bibinfo
  {author} {\bibfnamefont {M.~B.}\ \bibnamefont {Elliott}}, \ and\ \bibinfo
  {author} {\bibfnamefont {S.}~\bibnamefont {Wehner}},\ }\bibfield  {title}
  {\enquote {\bibinfo {title} {{Local quantum measurement and no-signaling
  imply quantum correlations.}}}\ }\href {\doibase
  10.1103/PhysRevLett.104.140401} {\bibfield  {journal} {\bibinfo  {journal}
  {Physical Review Letters}\ }\textbf {\bibinfo {volume} {104}},\ \bibinfo
  {pages} {140401} (\bibinfo {year} {2010})}\BibitemShut {NoStop}%
\bibitem [{\citenamefont {Allcock}\ \emph {et~al.}(2009)\citenamefont
  {Allcock}, \citenamefont {Brunner}, \citenamefont {Pawlowski},\ and\
  \citenamefont {Scarani}}]{Allcock2009}%
  \BibitemOpen
  \bibfield  {author} {\bibinfo {author} {\bibfnamefont {J.}~\bibnamefont
  {Allcock}}, \bibinfo {author} {\bibfnamefont {N.}~\bibnamefont {Brunner}},
  \bibinfo {author} {\bibfnamefont {M.}~\bibnamefont {Pawlowski}}, \ and\
  \bibinfo {author} {\bibfnamefont {V.}~\bibnamefont {Scarani}},\ }\bibfield
  {title} {\enquote {\bibinfo {title} {{Recovering part of the boundary between
  quantum and nonquantum correlations from information causality}},}\ }\href
  {\doibase 10.1103/PhysRevA.80.040103} {\bibfield  {journal} {\bibinfo
  {journal} {Physical Review A}\ }\textbf {\bibinfo {volume} {80}},\ \bibinfo
  {pages} {040103} (\bibinfo {year} {2009})}\BibitemShut {NoStop}%
\bibitem [{\citenamefont {Barnum}\ \emph {et~al.}(2006)\citenamefont {Barnum},
  \citenamefont {Barrett}, \citenamefont {Leifer},\ and\ \citenamefont
  {Wilce}}]{Barnum2006}%
  \BibitemOpen
  \bibfield  {author} {\bibinfo {author} {\bibfnamefont {H.}~\bibnamefont
  {Barnum}}, \bibinfo {author} {\bibfnamefont {J.}~\bibnamefont {Barrett}},
  \bibinfo {author} {\bibfnamefont {M.}~\bibnamefont {Leifer}}, \ and\ \bibinfo
  {author} {\bibfnamefont {A.}~\bibnamefont {Wilce}},\ }\bibfield  {title}
  {\enquote {\bibinfo {title} {{Cloning and Broadcasting in Generic
  Probabilistic Theories}},}\ }\href {http://arxiv.org/abs/quant-ph/0611295} {\
   (\bibinfo {year} {2006})},\ \Eprint {http://arxiv.org/abs/0611295}
  {arXiv:0611295 [quant-ph]} \BibitemShut {NoStop}%
\bibitem [{\citenamefont {Tylec}\ and\ \citenamefont
  {Ku{\'{s}}}(2015)}]{tylec2015}%
  \BibitemOpen
  \bibfield  {author} {\bibinfo {author} {\bibfnamefont {T.~I.}\ \bibnamefont
  {Tylec}}\ and\ \bibinfo {author} {\bibfnamefont {M.}~\bibnamefont
  {Ku{\'{s}}}},\ }\bibfield  {title} {\enquote {\bibinfo {title}
  {{Non-signaling boxes and quantum logics}},}\ }\href {\doibase
  10.1088/1751-8113/48/50/505303} {\bibfield  {journal} {\bibinfo  {journal}
  {Journal of Physics A}\ }\textbf {\bibinfo {volume} {48}},\ \bibinfo {pages}
  {505303} (\bibinfo {year} {2015})}\BibitemShut {NoStop}%
\bibitem [{\citenamefont {Tylec}, \citenamefont {Ku{\'{s}}},\ and\
  \citenamefont {Krajczok}(2015)}]{tylec2015-2}%
  \BibitemOpen
  \bibfield  {author} {\bibinfo {author} {\bibfnamefont {T.~I.}\ \bibnamefont
  {Tylec}}, \bibinfo {author} {\bibfnamefont {M.}~\bibnamefont {Ku{\'{s}}}}, \
  and\ \bibinfo {author} {\bibfnamefont {J.}~\bibnamefont {Krajczok}},\
  }\bibfield  {title} {\enquote {\bibinfo {title} {{Non-signaling theories and
  generalized probability}},}\ }\href {http://arxiv.org/abs/1512.02457}
  {\bibfield  {journal} {\bibinfo  {journal} {submitted to IJTP}\ } (\bibinfo
  {year} {2015})},\ \Eprint {http://arxiv.org/abs/1512.02457}
  {arXiv:1512.02457} \BibitemShut {NoStop}%
\bibitem [{\citenamefont {Pt{\'{a}}k}\ and\ \citenamefont
  {Pulmannov{\'{a}}}(1991)}]{ptak1991orthomodular}%
  \BibitemOpen
  \bibfield  {author} {\bibinfo {author} {\bibfnamefont {P.}~\bibnamefont
  {Pt{\'{a}}k}}\ and\ \bibinfo {author} {\bibfnamefont {S.}~\bibnamefont
  {Pulmannov{\'{a}}}},\ }\href@noop {} {\emph {\bibinfo {title} {{Orthomodular
  Structures as Quantum Logics: Intrinsic Properties, State Space and
  Probabilistic Topics}}}},\ Vol.~\bibinfo {volume} {44}\ (\bibinfo
  {publisher} {Springer},\ \bibinfo {year} {1991})\BibitemShut {NoStop}%
\bibitem [{\citenamefont {Namioka}\ and\ \citenamefont
  {Phelps}(1969)}]{namioka1969tensor}%
  \BibitemOpen
  \bibfield  {author} {\bibinfo {author} {\bibfnamefont {I.}~\bibnamefont
  {Namioka}}\ and\ \bibinfo {author} {\bibfnamefont {R.}~\bibnamefont
  {Phelps}},\ }\bibfield  {title} {\enquote {\bibinfo {title} {{Tensor products
  of compact convex sets}},}\ }\href@noop {} {\bibfield  {journal} {\bibinfo
  {journal} {Pacific Journal of Mathematics}\ }\textbf {\bibinfo {volume}
  {31}},\ \bibinfo {pages} {469--480} (\bibinfo {year} {1969})}\BibitemShut
  {NoStop}%
\bibitem [{\citenamefont {Jarchow}(1981)}]{Jarchow1981}%
  \BibitemOpen
  \bibfield  {author} {\bibinfo {author} {\bibfnamefont {H.}~\bibnamefont
  {Jarchow}},\ }\href {\doibase 10.1007/978-3-322-90559-8} {\emph {\bibinfo
  {title} {{Locally Convex Spaces}}}},\ Mathematische Leitf{\"{a}}den\
  (\bibinfo  {publisher} {Vieweg+Teubner Verlag},\ \bibinfo {address}
  {Wiesbaden},\ \bibinfo {year} {1981})\BibitemShut {NoStop}%
\bibitem [{\citenamefont {Pulmannov{\'{a}}}(1983)}]{pulmannova1983}%
  \BibitemOpen
  \bibfield  {author} {\bibinfo {author} {\bibfnamefont {S.}~\bibnamefont
  {Pulmannov{\'{a}}}},\ }\bibfield  {title} {\enquote {\bibinfo {title}
  {{Coupling of Quantum Logics}},}\ }\href {\doibase 10.1007/BF02114666}
  {\bibfield  {journal} {\bibinfo  {journal} {International Journal of
  Theoretical Physics}\ }\textbf {\bibinfo {volume} {22}},\ \bibinfo {pages}
  {837--850} (\bibinfo {year} {1983})}\BibitemShut {NoStop}%
\bibitem [{\citenamefont {Pt{\'{a}}k}\ and\ \citenamefont
  {Pulmannov{\'{a}}}(2007)}]{pulmanova2007}%
  \BibitemOpen
  \bibfield  {author} {\bibinfo {author} {\bibfnamefont {P.}~\bibnamefont
  {Pt{\'{a}}k}}\ and\ \bibinfo {author} {\bibfnamefont {S.}~\bibnamefont
  {Pulmannov{\'{a}}}},\ }\bibfield  {title} {\enquote {\bibinfo {title}
  {{Quantum logics as underlying strutures of generalized probability
  theory}},}\ }in\ \href@noop {} {\emph {\bibinfo {booktitle} {Handbook of
  Quantum Logic and Quantum Structures}}}\ (\bibinfo  {publisher} {Elsevier},\
  \bibinfo {year} {2007})\ pp.\ \bibinfo {pages} {147--213}\BibitemShut
  {NoStop}%
\bibitem [{\citenamefont {Varadarajan}(2007)}]{varadarajan2007geometry}%
  \BibitemOpen
  \bibfield  {author} {\bibinfo {author} {\bibfnamefont {V.~S.}\ \bibnamefont
  {Varadarajan}},\ }\href@noop {} {\emph {\bibinfo {title} {{Geometry of
  quantum theory}}}}\ (\bibinfo  {publisher} {Springer Science {\&} Business
  Media},\ \bibinfo {year} {2007})\BibitemShut {NoStop}%
\bibitem [{\citenamefont {Aerts}\ and\ \citenamefont
  {Daubechies}(1978)}]{aerts1978physical}%
  \BibitemOpen
  \bibfield  {author} {\bibinfo {author} {\bibfnamefont {D.}~\bibnamefont
  {Aerts}}\ and\ \bibinfo {author} {\bibfnamefont {I.}~\bibnamefont
  {Daubechies}},\ }\bibfield  {title} {\enquote {\bibinfo {title} {{Physical
  justification for using the tensor product to describe two quantum systems as
  one joint system}},}\ }\href@noop {} {\bibfield  {journal} {\bibinfo
  {journal} {Helv. Phys. Acta}\ }\textbf {\bibinfo {volume} {51}} (\bibinfo
  {year} {1978})}\BibitemShut {NoStop}%
\bibitem [{\citenamefont {Dvurecenskij}\ and\ \citenamefont
  {Pulmannov{\'{a}}}(2000)}]{dvurecenskij2000new}%
  \BibitemOpen
  \bibfield  {author} {\bibinfo {author} {\bibfnamefont {A.}~\bibnamefont
  {Dvurecenskij}}\ and\ \bibinfo {author} {\bibfnamefont {S.}~\bibnamefont
  {Pulmannov{\'{a}}}},\ }\href@noop {} {\emph {\bibinfo {title} {{New trends in
  quantum structures}}}},\ Vol.~\bibinfo {volume} {1}\ (\bibinfo {year}
  {2000})\BibitemShut {NoStop}%
\bibitem [{\citenamefont {Matolcsi}(1975)}]{matolcsi1975tensor}%
  \BibitemOpen
  \bibfield  {author} {\bibinfo {author} {\bibfnamefont {T.}~\bibnamefont
  {Matolcsi}},\ }\bibfield  {title} {\enquote {\bibinfo {title} {{Tensor
  product of Hilbert lattices and free orthodistributive product of
  orthomodular lattices}},}\ }\href@noop {} {\bibfield  {journal} {\bibinfo
  {journal} {Acta Scientiarum Mathematicarum}\ }\textbf {\bibinfo {volume}
  {37}},\ \bibinfo {pages} {263--272} (\bibinfo {year} {1975})}\BibitemShut
  {NoStop}%
\bibitem [{\citenamefont {Pulmannov{\'{a}}}(1985)}]{pulmannova1985tensor}%
  \BibitemOpen
  \bibfield  {author} {\bibinfo {author} {\bibfnamefont {S.}~\bibnamefont
  {Pulmannov{\'{a}}}},\ }\bibfield  {title} {\enquote {\bibinfo {title}
  {{Tensor product of quantum logics}},}\ }\href {\doibase 10.1063/1.526784}
  {\bibfield  {journal} {\bibinfo  {journal} {Journal of Mathematical Physics}\
  }\textbf {\bibinfo {volume} {26}},\ \bibinfo {pages} {1--5} (\bibinfo {year}
  {1985})}\BibitemShut {NoStop}%
\end{thebibliography}%

\end{document}